\documentclass[11pt]{article}
\usepackage{amsmath,amsthm,amsfonts,amssymb,comment}
\renewcommand{\=}{\doteq}
\newcommand{\p}{\partial}

\newcommand{\RR}{\mathbb{R}}
\newcommand{\X}{\mathcal{X}}
\newcommand{\al}{\alpha}

\newcommand{\apr}{\approx}

\newtheorem{thm}{Theorem}[section]
\theoremstyle{definition}
 \newtheorem{defn}[thm]{Definition}
\theoremstyle{definition}

\theoremstyle{definition}
 \newtheorem{rem}[thm]{Remark}
\theoremstyle{definition}

\numberwithin{equation}{section}

\begin{document}

\title{\huge\bf How to compose Lagrangian?}
\author{\Large Eugen Paal and J\"uri Virkepu
\\  \\
Tallinn University of Technology
\\
Ehitajate tee 5, 19086 Tallinn, Estonia
\\  \\
E-mails: eugen.paal@ttu.ee and  jvirkepu@staff.ttu.ee
 }
\date{}
\maketitle
\thispagestyle{empty}

\begin{abstract}

A method for constructing Lagrangians for the Lie transformation groups is explained.
As examples, the Lagrangians for real plane rotations and affine transformations of the real line are constructed.
\end{abstract}

\section{Introduction and outline of the paper}

It is a well-known problem in physics and mechanics how to construct
Lagrangians for mechanical systems via their equations of motion.
This \emph{inverse variational problem} has been investigated for some types of
equations of motion in \cite{l99}.

In \cite{bpv06, pv05}, the plane rotation group $SO(2)$ was
considered as a toy model of the Hamilton-Dirac mechanics with
constraints. By introducing a Lagrangian in a particular form,
canonical formalism for $SO(2)$ was developed. The crucial idea of
this approach is that the Euler-Lagrange and Hamilton canonical
equations must in a sense coincide with Lie equations of the Lie
transformation group.

In this paper, the method for constructing such a Lagrangian is
proposed. It is shown, how it is possible to find the Lagrangian,
based on the Lie equations of the Lie transformation group.

By composing a Lagrangian,  it is possible to describe given Lie transformation
group as a mechanical system and to develop the corresponding Lagrange and Hamilton formalisms.

\section{General method for constructing Lagrangians}

Let $G$ be an $r$-parametric Lie group with unit $e\in G$ and let
$g^{i}$ ($i=1,\dots,r$) denote the  local coordinates of an element
$g\in G$ from the vicinity of $e$. Let  $\X$ be an $n$-dimensional
manifold and denote the local coordinates of $X\in\X$ by $X^{\al}$
($\al=1,\dots,n$). Consider a (left) differentiable action of $G$ on
$\X$ given by
\[
\X\ni \quad X'=S_g X\quad\in \X
\]
Let $gh$ denote the \emph{multiplication} of $G$. Then
\[
S_g S_h = S_{gh},\qquad g,h\in G
\]
By introducing the \emph{auxiliary functions} $u^i_j$ and $S^\al_j$
by
\begin{align*}
(gh)^i&\=h^i+u^i_j(h)g^j+\dots\\
(S_gX)^\al&\=X^\al+S^\al_j(X)g^j+\dots
\end{align*}
the Lie equations read
\[
\varphi^{\al}_{j}(X;g)\=u^s_j(g)\dfrac{\p (S_g X)^\al}{\p
g^s}-S^{\al}_j(S_gX)=0
\]
The expressions $\varphi^{\al}_{j}$ are said to be \emph{constraints} for the 
Lie transformation group ($\X,G$). Then we search for such a vector Lagrangian
$\mathbf{L}\doteq(L_{1},\ldots,L_{r})$ with components
\[
L_{k}
\doteq
\sum_{\al=1}^{n}\sum_{s=1}^{r}\lambda_{k\al}^{s}\varphi_s^{\al},
\quad k=1,2,\ldots, r
\]
and such \emph{Lagrange multipliers} $\lambda_{k\al}^{s}$ that
the Euler-Lagrange equations in a sense coincide with the Lie
equations.

The notion of a vector Lagrangian was introduced and developed in \cite{f01,s86}.

\begin{defn}[weak equality]
The functions  $A$ and $B$ are called \emph{weakly equal}, if
\[
(A-B)\Big|_{\varphi_j^{\al}=0}=0
\quad
\forall j=1,2,\ldots,r,
\quad
\forall \al =1,2,\ldots,n
\]
In this case we write $A\apr B$.
\end{defn}

By denoting
\[
X'^{\al}_{i}\doteq\frac{\p X'^{\al}}{\p g^{i}}
\]
the conditions for the Lagrange multipliers read as the  
\emph{weak Euler-Lagrange equations}
\[
L_{k\al}\doteq\frac{\p L_{k}}{\p X'^{\al}}
-\sum_{i=1}^{r}\frac{\p}{\p g^{i}}\frac{\p L_{k}}{\p
X'^{\al}_i}\approx 0
\]
Finally, one must check by direct calculations that the Euler-Lagrange equations 
$L_{k\alpha}=0$ imply the Lie equations of the Lie transformation group.

\section{Lagrangian for $SO(2)$}

First consider the 1-parameter Lie transformation group $SO(2)$,
\emph{the rotation group of the real two-plane} $\RR^{2}$. In this
case $n=2$ and $r=1$. Rotation of the plane $\RR^{2}$ by an angle
$g\in\RR$ is given by the transformation
\[
\begin{cases}
(S_gX)^{1}=X'^{1}=X'^{1}(X^{1},X^{2},g)\=X^{1}\cos g-X^{2}\sin g\\
(S_gX)^{2}=X'^{2}=X'^{2}(X^{1},X^{2},g)\=X^{1}\sin g+X^{2}\cos g
\end{cases}
\]
We consider the rotation angle $g$ as a dynamical variable and the
functions $X'^{1}$ and $X'^{2}$ as \emph{field variables} for the
plane rotation group $SO(2)$.

Denote
\[
\dot{X}'^{\al}\doteq\frac{\p X'^{\alpha}}{\p g}
\]

The \emph{infinitesimal coefficients} of the transformation are
\[
\begin{cases}
S^1(X^{1},X^{2})\=\dot{X}'^{1}(X^{1},X^{2},e)=-X^{2}\\
S ^2(X^{1},X^{2})\=\dot{X}'^{2}(X^{1},X^{2},e)=X^{1}
\end{cases}
\]
and the Lie equations read
\[
\begin{cases}
\dot{X}'^{1}=S^1(X'^{1},X'^{2})=-X'^{2}\\
\dot{X}'^{2}=S^2(X'^{1},X'^{2})=X'^{1}
\end{cases}
\]
Rewrite the Lie equations in implicit form as follows:
\[
\begin{cases}
\varphi_1^{1}\doteq\dot{X}'^{1} +X'^{2}=0\\
\varphi_1^{2}\doteq\dot{X}'^{2} -X'^{1}=0
\end{cases}
\]
We search a Lagrangian of $SO(2)$ in the form
\[
L_{1} =\sum_{\al
=1}^{2}\sum_{s=1}^{1}\lambda_{1\al}^{s}\varphi_s^{\al}=\lambda_{11}^{1}\varphi_1^{1}+\lambda_{12}^{1}\varphi_1^{2}
\]
It is more convenient to rewrite it as follows:
\[
L\doteq\lambda_1\varphi^{1}+\lambda_2\varphi^{2}
\]
where the Lagrange multipliers $\lambda_1$ and $\lambda_2$ are to be
found from the weak Euler-Lagrange equations
\[
\dfrac{\p L}{\p X'^{1}}-\dfrac{\p}{\p g}\dfrac{\p L}{\p
\dot{X}'^{1}}\approx 0, \qquad \dfrac{\p L}{\p X'^{2}}-\dfrac{\p}{\p
g}\dfrac{\p L}{\p\dot{X}'^{2}}\approx 0
\]
Calculate
\begin{align*}
\dfrac{\p L}{\p X'^1} &=\dfrac{\p}{\p X'^1}
     \left[\lambda_1(\dot{X}'^{1}+X'^2)+\lambda_2(\dot{X}'^{2}-X'^1)\right]\\&=\dfrac{\p\lambda_1}{\p X'^1}\varphi^{1}+
     \dfrac{\p\lambda_2}{\p X'^1}\varphi^{2}-\lambda_2\approx -\lambda_2\\
\dfrac{\p L}{\p\dot{X}'^{1}} &=\dfrac{\p}{\p \dot{X}'^{1}}
     \left[\lambda_1(\dot{X}'^{1}+X'^2)+\lambda_2(\dot{X}'^{2}-X'^1)\right]\approx\lambda_1\\
\dfrac{\p}{\p g}\dfrac{\p L}{\p\dot{X}'^{1}}&=\dfrac{\p\lambda_1}{\p
g}=\dfrac{\p\lambda_1}{\p X'^1}\dot{X}'^{1} +\dfrac{\p\lambda_1}{\p
X'^2}\dot{X}'^{2}\approx -\dfrac{\p\lambda_1}{\p X'^1}X'^2
+\dfrac{\p\lambda_1}{\p X'^2}X'^1
\end{align*}
from which it follows
\[
\dfrac{\p L}{\p X'^1}-\dfrac{\p}{\p g}\dfrac{\p
L}{\p\dot{X}'^{1}}\approx 0\qquad\Longleftrightarrow\qquad
-\lambda_2+\dfrac{\p\lambda_1}{\p X'^1}X'^2 -\dfrac{\p\lambda_1}{\p
X'^2}X'^1\approx 0
\]
Analogously calculate
\begin{align*}
\dfrac{\p L}{\p X'^2} &=\dfrac{\p}{\p X'^2}
     \left[\lambda_1(\dot{X}'^{1}+X'^2)+\lambda_2(\dot{X}'^{2}-X'^1)\right]\\&=\dfrac{\p\lambda_1}{\p X'^2}\varphi^{1}+
     \dfrac{\p\lambda_2}{\p X'^2}\varphi^{2}+\lambda_1\approx \lambda_1\\
\dfrac{\p L}{\p\dot{X}'^{2}} &=\dfrac{\p}{\p \dot{X}'^{2}}
     \left[\lambda_1(\dot{X}'^{1}+X'^2)+\lambda_2(\dot{X}'^{2}-X'^1)\right]\approx\lambda_2\\
\dfrac{\p}{\p g}\dfrac{\p L}{\p\dot{X}'^{2}}&=\dfrac{\p\lambda_2}{\p
g}=\dfrac{\p\lambda_2}{\p X'^1}\dot{X}'^{1} +\dfrac{\p\lambda_2}{\p
X'^2}\dot{X}'^{2}\approx -\dfrac{\p\lambda_2}{\p X'^1}X'^2
+\dfrac{\p\lambda_2}{\p X'^2}X'^1
\end{align*}
from which it follows
\[
\dfrac{\p L}{\p X'^2}-\dfrac{\p}{\p g}\dfrac{\p
L}{\p\dot{X}'^{2}}\approx 0\qquad\Longleftrightarrow\qquad
\lambda_1+\dfrac{\p\lambda_2}{\p X'^1}X'^2 -\dfrac{\p\lambda_2}{\p
X'^2}X'^1\approx 0
\]
So the calculations imply the following system of differential
equations for the Lagrange multipliers:
\[
\begin{cases}
-\dfrac{\p\lambda_1}{\p
X'^1}X'^2 +\dfrac{\p\lambda_1}{\p X'^2}X'^1\approx -\lambda_2\\
-\dfrac{\p\lambda_2}{\p X'^1}X'^2 +\dfrac{\p\lambda_2}{\p
X'^2}X'^1\approx \lambda_1
\end{cases}
\]
We are not searching for the general solution for this system of
partial differential equations, but the Lagrange multipliers are
supposed to be a linear combination of the field variables $X'^1$
and $X'^2$,
\[
\begin{cases}
\lambda_1\doteq\alpha_1X'^1+\alpha_2X'^2\\
\lambda_2\doteq\beta_1X'^1+\beta_2X'^2, \qquad
\alpha_1,\alpha_2,\beta_1,\beta_2\in\mathbb{R}
\end{cases}
\]
By using these expressions, one has
\[
\begin{cases}
-\alpha_1X'^2+\alpha_2X'^1\approx -\beta_1X'^1-\beta_2X'^2\\
-\beta_1X'^2+\beta_2X'^1\approx \alpha_1X'^1+\alpha_2X'^2
\end{cases}
\quad\Longleftrightarrow\quad
\begin{cases}
(\alpha_2+\beta_1)X'^1+(\beta_2-\alpha_1)X'^2\approx 0\\
(\beta_2-\alpha_1)X'^1-(\alpha_2+\beta_1)X'^2\approx 0
\end{cases}
\]
This is a homogeneous system of two linear equations of four
unknowns $\alpha_1,\alpha_2,\beta_1,\beta_2$. The system is
satisfied, if
\[
\begin{cases}
\alpha_2+\beta_1=0\\
\beta_2-\alpha_1=0
\end{cases}
\qquad\Longleftrightarrow\qquad
\begin{cases}
\beta_1=-\alpha_2\\
\beta_2=\alpha_1
\end{cases}
\]
The parameters $\alpha_1,\alpha_2$ are free. Thus,
\[
\begin{cases}
\lambda_1=\alpha_1X'^1+\alpha_2X'^2\\
\lambda_2=-\alpha_2X'^1+\alpha_1X'^2
\end{cases}
\]
and the desired Lagrangian for $SO(2)$ reads
\begin{equation}
\label{E-L eq1}
L=\alpha_1(X'^1\dot{X}'^{1}
+X'^2\dot{X}'^{2})+\alpha_2\left[X'^2\dot{X}'^{1}
+(X'^2)^{2}-X'^1\dot{X}'^{2} +(X'^1)^{2}\right]
\end{equation}
with free real parameters $\alpha_1,\alpha_2$.
Thus we can propose the

\begin{thm}
The Euler-Lagrange equations for the Lagrangian (\ref{E-L eq1})  coincide with 
the  Lie equations of $SO(2)$.
\end{thm}

\begin{proof}
Calculate
\begin{align*}
\dfrac{\p L}{\p X'^1} &=\dfrac{\p}{\p X'^1}
     \left[\alpha_1(X'^1\dot{X}'^{1} +X'^2\dot{X}'^{2})+
     \alpha_2\left(X'^2\dot{X}'^{1} +(X'^2)^{2}-X'^1\dot{X}'^{2}
     +(X'^1)^{2}\right)\right]\\
     &=\alpha_1\dot{X}'^{1}-\alpha_2\dot{X}'^{2} +2\alpha_2X'^1\\
\dfrac{\p L}{\p\dot{X}'^{1}} &=\dfrac{\p}{\p\dot{X}'^{1}}
\left[\alpha_1(X'^1\dot{X}'^{1} +X'^2\dot{X}'^{2})+
     \alpha_2\left(X'^2\dot{X}'^{1} +(X'^2)^{2}-X'^1\dot{X}'^{2}
     +(X'^1)^{2}\right)\right]\\&=\alpha_1X'^1+\alpha_2X'^2
\quad\Longrightarrow\quad
\dfrac{\p}{\p g}\dfrac{\p L}{\p\dot{X}'^{1}}=\alpha_1\dot{X}'^{1}+\alpha_2\dot{X}'^{2}
\end{align*}
from which it follows
\[
\dfrac{\p L}{\p X'^1}-\dfrac{\p}{\p g}\dfrac{\p L}{\p\dot{X}'^{1}}=0
\quad\Longleftrightarrow\quad 2\alpha_2X'^1-2\alpha_2\dot{X}'^{2}=0
\quad\Longleftrightarrow\quad \dot{X}'^{2}=X'^1
\]
Analogously calculate
\begin{align*}
\dfrac{\p L}{\p X'^2} &=\dfrac{\p}{\p X'^2}
     \left[\alpha_1(X'^1\dot{X}'^{1} +X'^2\dot{X}'^{2})+
     \alpha_2\left(X'^2\dot{X}'^{1} +(X'^2)^{2}-X'^1\dot{X}'^{2}
     +(X'^1)^{2}\right)\right]\\
     &=\alpha_2\dot{X}'^{1}+2\alpha_2X'^2+\alpha_1\dot{X}'^2\\
\dfrac{\p L}{\p\dot{X}'^{2}} &=\dfrac{\p}{\p\dot{X}'^{2}}
\left[\alpha_1(X'^1\dot{X}'^{1} +X'^2\dot{X}'^{2})+
     \alpha_2\left(X'^2\dot{X}'^{1} +(X'^2)^{2}-X'^1\dot{X}'^{2}
     +(X'^1)^{2}\right)\right]\\&=-\alpha_2X'^1+\alpha_1X'^2
\quad\Longrightarrow\quad
\dfrac{\p}{\p g}\dfrac{\p L}{\p\dot{X}'^{2}}=-\alpha_2\dot{X}'^{1}+\alpha_1\dot{X}'^{2}
\end{align*}
from which it follows
\[
\dfrac{\p L}{\p X'^2}-\dfrac{\p}{\p g}\dfrac{\p L}{\p\dot{X}'^{2}}=0
\quad\Longleftrightarrow\quad 2\alpha_2\dot{X}'^{1}+2\alpha_2X'^2=0
\quad\Longleftrightarrow\quad \dot{X}'^{1}=-X'^2
\qedhere
\]

\end{proof}

\section{Physical interpretation}

While the Lagrangian $L$ of $SO(2)$ contains two free parameters
$\alpha_1,\alpha_2,$ particular forms of it can be found taking
into account physical considerations.
In particular, if $\alpha_1=0$ and $\alpha_2=-1/2$, then the
Lagrangian of $SO(2)$ reads
\[
L(X'^1,X'^2,\dot{X}'^{1},\dot{X}'^{2})
\=\dfrac{1}{2}(X'^1\dot{X}'^{2}-\dot{X}'^{1}X'^2)
-\dfrac{1}{2}\left[(X'^1)^{2}+(X'^2)^{2}\right]
\]
By using the Lie equations one can easily check that
\[
X'^1\dot{X}'^{2}-\dot{X}'^{1}
X'^2=(\dot{X}'^{1})^{2}+(\dot{X}'^{2})^{2}
\]
The function
\[
T\doteq
\dfrac{1}{2}\left[(\dot{X}'^{1})^{2}+(\dot{X}'^{2})^{2}\right]
\]
is the \emph{kinetic energy} of a moving point $(X'^1,X'^2)\in
\RR^{2}$, meanwhile
\[
l\doteq X'^1\dot{X}'^{2}-\dot{X}'^{1} X'^2
\]
is its \emph{kinetic momentum} with respect to origin $(0,0)\in
\RR^{2}$.

This relation has a simple explanation in the kinematics of a rigid
body \cite{g53}.
The kinetic energy of a point can be represented via its kinetic
momentum as follows:
\[
\frac{1}{2}\left[(\dot{X}^{1})^{2}+(\dot{X}'^{2})^{2}\right]
=T=\frac{l}{2}
=\frac{1}{2}\left[X'^1\dot{X}'^{2}-\dot{X}^{1}X'^2\right]
\]
Thus we can conclude, that for the given Lie equations (that is, on the extremals)
of $SO(2)$ the Lagrangian $L$ gives rise to a Lagrangian of a pair of harmonic oscillators.

\section{Lagrangian for the affine transformations of the line}

Now consider the affine transformations of the real line. The latter
may be represented by
\[
\begin{cases}
X'^1=X'^1(X^1,X^2,g^{1},g^2)\doteq g^{1}X^{1}+g^{2}\\
X'^2=X'^2(X^1,X^2,g^1,g^2)\doteq 1,
\qquad\qquad 0\neq g^{1},g^{2}\in\mathbb{R}
\end{cases}
\]
Thus $r=2$ and $n=2$.
Denote
\[
e\doteq(1,0),\quad g^{-1}\doteq \frac{1}{g^{1}}(1,-g^{2})
\]
First, find the multiplication rule
\begin{align*}
(X'')^{1}&\doteq(X'^{1})'=S_{gh}X^{1}=S_g(S_hX^{1})=S_g(h^{1}X^{1}+h^{2})\\&=g^{1}(h^{1}X^{1}+h^{2})+g^{2}
=(g^{1}h^{1})X^{1}+(g^{1}h^{2}+g^{2})
\end{align*}
Calculate the infinitesimal coefficients
\begin{align*}
S^1_1(X^{1},X^{2})&\doteq\left.X'^{1}_1\right|_{g=e}=X_1\\
S^1_2(X^{1},X^{2})&\doteq\left.X'^{1}_2\right|_{g=e}=1\\
S^2_1(X^{1},X^{2})&\doteq\left.X'^{2}_1\right|_{g=e}=0\\
S^2_2(X^{1},X^{2})&\doteq\left.X'^{2}_2\right|_{g=e}=0
\end{align*}
and auxiliary functions
\begin{align*}
u_1^{1}(g)&\doteq\left.\frac{\p (S_{gh}X)^{1}}{\p g^{1}}\right|_{h=g^{-1}}=\frac{1}{g^{1}}\\
u_2^{1}(g)&\doteq\left.\frac{\p (S_{gh}X)^{1}}{\p g^{2}}\right|_{h=g^{-1}}=0\\
u_1^{2}(g)&\doteq\left.\frac{\p (S_{gh}X)^{2}}{\p g^{1}}\right|_{h=g^{-1}}=-\frac{g^{2}}{g^{1}}\\
u_2^{2}(g)&\doteq\left.\frac{\p (S_{gh}X)^{2}}{\p
g^{2}}\right|_{h=g^{-1}}=1
\end{align*}
Next, write Lie equations and find constraints
\[
\left\{
                                          \begin{array}{ll}
                                             X'^{1}_1=\frac{1}{g^{1}}X'^{1}-\frac{g^{2}}{g^{1}} \\
                                             X'^{1}_2=1 \\
                                             X'^{2}_1=0 \\
                                             X'^{2}_2=0
                                           \end{array}
                                         \right.
\quad\Longleftrightarrow\quad \left\{
                                          \begin{array}{ll}
                                             \varphi_1^{1}\doteq X'^{1}_1-\frac{1}{g^{1}}X'^{1}-\frac{g^{2}}{g^{1}} \\
                                            \varphi_2^{1}\doteq X'^{1}_2-1 \\
                                             \varphi_1^{2}\doteq X'^{2}_1\\
                                             \varphi_2^{2}\doteq X'^{2}_2
                                           \end{array}
                                         \right.
\]
We search for a vector Lagrangian $\mathbf{L}=(L_{1},L_{2})$ as
follows:
\begin{align*}
L_{k}&=\sum_{\al
=1}^{2}\sum_{s=1}^{2}\lambda_{k\al}^{s}\varphi_s^{\al} =\lambda_{k
1}^{1}\varphi_1^{1}+\lambda_{k 1}^{2}\varphi_2^{1}+\lambda_{k
2}^{1}\varphi_1^{2}+\lambda_{k 2}^{2}\varphi_2^{2}\\&=\lambda_{k
1}^{1}\left(X'^{1}_1-\frac{1}{g^{1}}X'^{1}-\frac{g^{2}}{g^{1}}\right)
+\lambda_{k 1}^{2}\left(X'^{1}_2-1\right)+\lambda_{k
2}^{1}X'^{2}_1+\lambda_{k 2}^{2}X'^{2}_2,\quad k=1,2
\end{align*}
By substituting the Lagrange multipliers $\lambda_{k\al}^{s}$ into
the weak Euler-Lagrange equations
\[
\frac{\p L_{k}}{\p X'^{\al}}-\sum_{i=1}^{2}\frac{\p}{\p
g^{i}}\frac{\p L_{k}}{\p X'^{\al}_i}\approx 0
\]
we get the following PDE system
\begin{align*}
\begin{cases}
(X'^{1}-g^{2})\frac{\p\lambda_{k 1}^{1}}{\p X'^{1}}+g^{1}\frac{\p\lambda_{k 1}^{2}}{\p X'^{1}}+\lambda_{k 1}^{1}\approx 0 \\
(X'^{1}-g^{2})\frac{\p\lambda_{k 2}^{1}}{\p
X'^{1}}+g^{1}\frac{\p\lambda_{k 2}^{2}}{\p X'^{1}}\approx 0, \quad
k=1,2
\end{cases}
\end{align*}
We find some particular solutions for this system. For example,
\[
k=1:
\begin{cases}
\lambda_{11}^{1}\doteq 0 \\
\lambda_{11}^{2}\doteq \psi_{11}^{2}(X'^2) \\
\lambda_{12}^{1}\doteq \psi_{12}^{1}(X'^2)\\
\lambda_{12}^{2}\doteq \psi_{12}^{2}(X'^2)
\end{cases}
\qquad\text{and}\qquad
k=2:
\begin{cases}
\lambda_{21}^{1}\doteq\psi_{21}^{1}(X'^2) \\
\lambda_{21}^{2}\doteq -\frac{X'^{1}}{g^{1}}\psi_{21}^{1}(X'^2) \\
\lambda_{22}^{1}\doteq 0\\
\lambda_{11}^{2}\doteq 0
\end{cases}
\]
with
$\psi_{21}^1(X'^2),\psi_{11}^2(X'^2),\psi_{12}^1(X'^2),\psi_{12}^2(X'^2)$
as arbitrary real valued functions of $X'^2$.

Thus we can define the Lagrangian $\mathbf{L}=(L_1,L_2)$ with
\begin{align}
\label{E-L eq2}
\begin{cases}
L_1=\psi_{11}^2(X'^2)(X'^{1}_2-1)+\psi_{12}^1(X'^2)X'^{2}_1+\psi_{12}^2(X'^2)X'^{2}_2\\
L_2=\psi_{21}^1(X'^2)\left(X'^{1}_1-\frac{1}{g^{1}}X'^{1}+\frac{g^{2}}{g^{1}}\right)
\end{cases}
\end{align}
and propose the

\begin{thm}
The Euler-Lagrange equations for the vector Lagrangian  $\mathbf{L}=(L_{1},L_{2})$ with
components (\ref{E-L eq2}) 
coincide with the Lie equations of the affine transformations of the real line.
\end{thm}

\begin{proof}
Calculate
\begin{align*}
\frac{\p L_1}{\p X'^{1}}&=\frac{\p}{\p
X'^{1}}\left[\psi_{11}^2(X'^2)(X'^{1}_2-1)+\psi_{12}^1(X'^2)X'^{2}_1+\psi_{12}^2(X'^2)X'^{2}_2\right]=0\\
\frac{\p}{\p g^{1}}\frac{\p L_1}{\p X'^{1}_1}&=\frac{\p}{\p
g^{1}}0=0\\
\frac{\p}{\p g^{2}}\frac{\p L_1}{\p
X'^{1}_2}&=\frac{\p\psi_{11}^2(X'^2)}{\p
g^{2}}=\frac{\p\psi_{11}^2(X'^2)}{\p X'^{2}}X'^{2}_2
\end{align*}
from which it follows
\[
\frac{\p L_1}{\p X'^{1}}-\sum_{i=1}^{2}\frac{\p}{\p g^{i}}\frac{\p
L_1}{\p X'^{1}_i}=0\quad\Longleftrightarrow\quad
\frac{\p\psi_{11}^2(X'^2)}{\p
X'^{2}}X'^{2}_2=0\quad\Longrightarrow\quad X'^{2}_2=0
\]
Analogously calculate
\begin{align*}
\frac{\p L_1}{\p X'^{2}}&=\frac{\p}{\p
X'^{2}}\left[\psi_{11}^2(X'^2)(X'^{1}_2-1)+\psi_{12}^1(X'^2)X'^{2}_1+\psi_{12}^2(X'^2)X'^{2}_2\right]\\&=
\frac{\p\psi_{11}^2(X'^2)}{\p
X'^{2}}(X'^{1}_2-1)+\frac{\p\psi_{12}^1(X'^2)}{\p
X'^{2}}X'^{2}_1+\frac{\p\psi_{12}^2(X'^2)}{\p
X'^{2}}X'^{2}_2\\
\frac{\p}{\p g^{1}}\frac{\p L_1}{\p
X'^{2}_1}&=\frac{\p\psi_{12}^1(X'^2)}{\p
g^{1}}=\frac{\p\psi_{12}^1(X'^2)}{\p
X'^{2}}X'^{2}_1\\
\frac{\p}{\p g^{2}}\frac{\p L_1}{\p
X'^{2}_2}&=\frac{\p\psi_{12}^2(X'^2)}{\p
g^{2}}=\frac{\p\psi_{12}^2(X'^2)}{\p X'^{2}}X'^{2}_2
\end{align*}
from which it follows
\[
\frac{\p L_1}{\p X'^{2}}-\sum_{i=1}^{2}\frac{\p}{\p g^{i}}\frac{\p
L_1}{\p X'^{2}_i}=0\quad\Longleftrightarrow\quad
\frac{\p\psi_{11}^2(X'^2)}{\p
X'^{2}}(X'^{1}_2-1)=0\quad\Longrightarrow\quad X'^{1}_2-1=0
\]
Now we differentiate the second component of the Lagrangian
$\mathbf{L}$. Calculate
\begin{align*}
\frac{\p L_2}{\p X'^{1}}&=\frac{\p}{\p
X'^{1}}\left[\psi_{21}^1(X'^2)\left(X'^{1}_1-\frac{1}{g^{1}}X'^{1}+\frac{g^{2}}{g^{1}}\right)-
\frac{1}{g^{1}}X'^{1}\psi_{21}^1(X'^2)(X'^{1}_2-1)\right]\\&=
-\frac{1}{g^{1}}\psi_{21}^1(X'^2)-\frac{1}{g^{1}}\psi_{21}^1(X'^2)X'^{1}_2+
\frac{1}{g^{1}}\psi_{21}^1(X'^2)=-\frac{1}{g^{1}}\psi_{21}^1(X'^2)X'^{1}_2\\
\frac{\p}{\p g^{1}}\frac{\p L_2}{\p
X'^{1}_1}&=\frac{\p\psi_{21}^1(X'^2)}{\p
g^{1}}=\frac{\p\psi_{21}^1(X'^2)}{\p
X'^{2}}X'^{2}_1\\
\frac{\p}{\p g^{2}}\frac{\p L_2}{\p X'^{1}_2}&=\frac{\p}{\p
g^{2}}\left(-\frac{1}{g^{1}}X'^{1}\psi_{21}^1(X'^2)\right)=-\frac{1}{g^{1}}\left(\psi_{21}^1(X'^2)X'^{1}_2+X'^{1}\frac{\p\psi_{21}^1(X'^2)}{\p
X'^{2}}X'^{2}_2\right)
\end{align*}
from which it follows
\begin{align*}
\frac{\p L_2}{\p X'^{1}}-\sum_{i=1}^{2}\frac{\p}{\p g^{i}}\frac{\p
L_2}{\p X'^{1}_i}=0\quad&\Longleftrightarrow\quad
\frac{\p\psi_{21}^1(X'^2)}{\p
X'^{2}}\left(X'^{2}_1-\frac{1}{g^{1}}X'^{1}X'^{2}_2\right)=0\\&\Longrightarrow\quad
X'^{2}_1-\frac{1}{g^{1}}X'^{1}X'^{2}_2=0
\end{align*}
Analogously calculate
\begin{align*}
\frac{\p L_2}{\p X'^{2}}&=\frac{\p}{\p
X'^{2}}\left[\psi_{21}^1(X'^2)\left(X'^{1}_1-\frac{1}{g^{1}}X'^{1}+\frac{g^{2}}{g^{1}}\right)-
\frac{1}{g^{1}}X'^{1}\psi_{21}^1(X'^2)(X'^{1}_2-1)\right]\\&=
\frac{\p\psi_{21}^1(X'^2)}{\p
X'^{2}}\left(X'^{1}_1-\frac{1}{g^{1}}X'^{1}+\frac{g^{2}}{g^{1}}\right)-
\frac{1}{g^{1}}\frac{\p\psi_{21}^1(X'^2)}{\p X'^{2}}X'^{1}\left(X'^{1}_2-1\right)\\
\frac{\p}{\p g^{1}}\frac{\p L_2}{\p X'^{2}_1}&=\frac{\p}{\p
g^{1}}0=0\\
\frac{\p}{\p g^{2}}\frac{\p L_2}{\p X'^{2}_2}&=\frac{\p}{\p
g^{2}}0=0
\end{align*}
from which it follows
\begin{align*}
\frac{\p L_2}{\p X'^{2}}-\sum_{i=1}^{2}\frac{\p}{\p g^{i}}\frac{\p
L_2}{\p X'^{2}_i}=0\quad&\Longleftrightarrow\quad
\frac{\p\psi_{21}^1(X'^2)}{\p
X'^{2}}\left(X'^{1}_1-\frac{1}{g^{1}}X'^{1}X'^{1}_2+\frac{g^{2}}{g^{1}}\right)=0\\&\Longrightarrow\quad
X'^{1}_1-\frac{1}{g^{1}}X'^{1}X'^{1}_2+\frac{g^{2}}{g^{1}}=0
\end{align*}
Thus, the Euler-Lagrange equations read
\[
\left\{
  \begin{array}{ll}
    X'^{2}_2=0 \\
    X'^{1}_2-1=0 \\
    X'^{2}_1-\frac{1}{g^{1}}X'^{1}X'^{2}_2=0\\
    X'^{1}_1-\frac{1}{g^{1}}X'^{1}X'^{1}_2+\frac{g^{2}}{g^{1}}=0
  \end{array}
\right.
\]
It can be easily verified, that the latter is equivalent to the
system of the Lie equations.
\end{proof}

\begin{rem}
While the Lagrangian $\mathbf{L}$ contains four arbitrary functions,
particular forms of it can be fixed by taking into account physical considerations.
\end{rem}

\section*{Acknowledgement}

The paper was in part supported by the Estonian Science Foundation, Grant 6912.

\end{document}